\documentclass[a4paper, 12pt]{article}

\usepackage{natbib}
\bibpunct{(}{)}{;}{a}{}{,} 

\usepackage{amsthm, amsmath, amssymb, mathrsfs, multirow, url}
\usepackage{graphicx} 
\usepackage{ifthen} 
\usepackage{amsfonts}
\usepackage[usenames,dvipsnames]{color}
\usepackage{fullpage}
\usepackage{booktabs}
\usepackage{caption}
\usepackage{subfigure}

\theoremstyle{plain} 
\newtheorem{theorem}{Theorem}
\newtheorem{corollary}{Corollary}

\theoremstyle{definition}

\theoremstyle{remark} 
\newtheorem{example}{Example}

\newcommand{\prob}{\mathsf{P}}

\renewcommand{\S}{\mathcal{S}}

\newcommand{\del}{\partial}

\renewcommand{\phi}{\varphi} 
\newcommand{\eps}{\varepsilon}

\newcommand{\avar}{\sigma_\alpha^2}
\newcommand{\evar}{\sigma_\eps^2}
\newcommand{\diag}{\mathrm{diag}}

\newcommand{\pl}{\mathsf{pl}}

\newcommand{\unif}{\mathsf{Unif}}
\newcommand{\nm}{\mathsf{N}}
\newcommand{\chisq}{\mathsf{ChiSq}}

\title{Exact prior-free probabilistic inference on the heritability coefficient in a linear mixed model}
\author{
Qianshun Cheng \quad Xu Gao \quad Ryan Martin \\
Department of Mathematics, Statistics, and Computer Science \\
University of Illinois at Chicago \\
{\tt (qcheng5, xgao21, rgmartin)@uic.edu} 
}
\date{\today}

\begin{document}

\maketitle

\begin{abstract}
Linear mixed-effect models with two variance components are often used when variability comes from two sources.  In genetics applications, variation in observed traits can be attributed to biological and environmental effects, and the heritability coefficient is a fundamental quantity that measures the proportion of total variability due to the biological effect.  We propose a new inferential model approach which yields exact prior-free probabilistic inference on the heritability coefficient.  In particular we construct exact confidence intervals and demonstrate numerically our method's efficiency compared to that of existing methods.

\smallskip

\end{abstract}

\section{Introduction}
\label{S:intro}

Normal linear mixed effects models are useful in a variety of biological, physical, and social scientific applications with variability coming from multiple sources; see \citet{khuri.sahai.1985} and \citet{searle.casella.mcculloch.1992}.  In this paper, we focus on the case of two variance components, and the general model can be written as 
\begin{equation}
\label{eq:mixed.effects}
Y = X \beta + Z \alpha + \eps,
\end{equation}
where $Y$ is a $n$-vector of response variables, $X$ and $Z$ are design matrices for the fixed and random effects, of dimension $n \times p$ and $n \times a$, respectively, $\beta$ is a $p$-vector of unknown parameters, $\alpha$ is a normal random $a$-vector with mean 0 and covariance matrix $\avar A$, and $\eps$ is a normal random $n$-vector with mean 0 and covariance matrix $\evar I_n$.  Here, $\sigma^2=(\avar, \evar)$ is the pair of variance components.  Assume $A$ is known and $X$ is of rank $p < n$.  The unknown parameters in this model are the $p$ fixed-effect coefficients $\beta$ and the two variance components $\sigma^2$, so the parameter space is $(p+2)$-dimensional.   

In biological applications, the quantities $\alpha$ and $\eps$ in \eqref{eq:mixed.effects} denote the genetic and environmental effects, respectively.  Given that ``a central question in biology is whether observed variation in a particular trait is due to environmental or biological factors'' \citep{visscher.natrev.2008}, the heritability coefficient, $\rho=\avar / (\avar + \evar)$, which represents the proportion of phenotypic variance attributed to variation in genotypic values, is a fundamentally important quantity.  Indeed, mixed-effect models and inference on the heritability coefficient has been applied recently in genome-wide association studies \citep{yang.natgen.2010, golan.rosset.2011}; see Section~\ref{S:discuss} for more on these applications.  

Given the importance of the heritability coefficient, there are a number of methods available to construct confidence intervals for $\rho$ or, equivalently, for the variance ratio $\psi=\avar / \evar$.  When the design is balanced, \citet{graybill1976} and \citet{searle.casella.mcculloch.1992} give a confidence intervals for $\psi$ and other quantities.  When the design is possibly unbalanced, as we assume here, the problem is more challenging; in particular, exact ANOVA-based confidence intervals generally are not available.  \citet{wald1940} gave intervals for $\psi$ in the unbalanced case, and subsequent contributions include \citet{harville.fenech.1985}, \citet{fenech.harville.1991}, \citet{lee.seely.1996}, and \citet{burch.iyer.1997}.  Bayesian \citep{wolfinger.kass.2000, gelman.bda.2004, gelman2006} and fiducial methods \citep{fisher1935a, e.hannig.iyer.2008, cisewski.hannig.2012} are also available. 

The focus in this paper is exact prior-free probabilistic inference on the heritability coefficient based on the recently proposed inferential model (IM) framework.  \citet{imbasics, imcond, immarg} give a detailed account of the general framework, along with comparisons to other related approaches, including fiducial \citep{hannig2009, hannig2012, fisher1973}.  The IM approach is driven by the specification of an association between unobservable auxiliary variables, the observable data, and the unknown parameters.  This association is followed up by using properly calibrated random sets to predict these auxiliary variables.  \citet{liu.martin.wire} argue that it is precisely this prediction of the auxiliary variables that sets the IM approach apart from fiducial, Dempster--Shafer \citep{dempster2008, shafer1976}, and structural inference \citep{fraser1968}.  The IM output is a plausibility function that can be used for inference on the parameter of interest.  In particular, the plausibility function yields confidence intervals with exact coverage \citep[e.g.,][]{randset}.  A brief overview of the general IM framework is given in Section~\ref{SS:review}.  A key feature of the IM approach, and the motivation for this work, is that, besides being exact, the IM approach's careful handling of uncertainty often leads to more efficient inference.   

A key to the construction of an efficient IM is to reduce the dimension of the auxiliary variable as much as possible and, for the heritability coefficient, there are several dimension reduction steps.  A first dimension reduction eliminates the nuisance parameter $\beta$.  These well-known calculations are reviewed in Section~\ref{SS:fixed.effect}.  The main technical contribution here is in Section~\ref{SS:construction}, where we employ a differential equation-driven technique to reduce the auxiliary variable dimension, leading to a conditional IM for the heritability coefficient.  In particular, this step allows us to reduce the dimension beyond that allowed by sufficiency.  A predictive random set is introduced in Section~\ref{SS:prs}, and in Section~\ref{SS:pi} we show that the corresponding plausibility function is properly calibrated and, therefore, the plausibility function-based confidence intervals are provably exact.  Sections~\ref{SS:simulation} and \ref{SS:real} illustrate the proposed method in simulated and real data examples.  The general message is that our proposed confidence intervals for the heritability coefficient are exact and tend to be more efficient compared to existing methods.

\section{Preliminaries}
\label{S:prelim}

\subsection{Overview of inferential models}
\label{SS:review}

The goal of the IM framework is to get valid prior-free probabilistic inference.  This is facilitated by first associating the observable data and unknown parameters to a set of unobservable auxiliary variables.  For example, the marginal distribution of $Y$ from \eqref{eq:mixed.effects} is 
\begin{equation}
\label{eq:Y.marginal}
Y \sim \nm_n(X\beta, \evar I_n + \avar Z A Z^\top), 
\end{equation}
which can be written in association form:
\begin{equation}
\label{eq:first.assoc}
Y = X\beta + (\evar I_n + \avar Z A Z^\top)^{1/2} \, U, \quad U \sim \nm_n(0, I_n). 
\end{equation}
That is, observable data $Y$ and unknown parameters $(\beta, \sigma^2)$ are associated with unobservable auxiliary variables $U$, in this case, a $n$-vector of independent standard normal random variables.  Given this association, the basic IM approach is to introduce a predictive random set for $U$ and combine with the observed value of $Y$ to get a plausibility function.  Roughly, this plausibility function assigns, to assertions about the parameter of interest, values in $[0,1]$ measuring the plausibility that the assertion is true.  This plausibility function can be used to design exact frequentist tests or confidence regions. 

\citet{imbasics} give a general description of this three-step IM construction and its properties; in Section~\ref{SS:construction} below we will flesh out these three steps, including choice of a predictive random set, for the variance components problem at hand.  The construction of exact plausibility intervals for the heritability coefficient is presented in Section~\ref{SS:pi}, along with a theoretical justification of the claimed exactness.  


Notice, in the association \eqref{eq:first.assoc}, that the auxiliary variable $U$ is, in general, of higher dimension than that of the parameter $(\beta,\sigma^2)$.  There are two reasons for trying to reduce the auxiliary variable dimension.  First, it is much easier to specify good predictive random sets when the auxiliary variable is of low dimension.  Second, inference is more efficient when only a relatively small number of auxiliary variables require prediction.  This dimension reduction is via a conditioning operation, and \citet{imcond} give a general explanation of the gains as well as a novel technique for carrying out this reduction.  We employ these techniques in Section~\ref{SS:construction} to get a dimension-reduced association for the heritability coefficient in the linear mixed model.  

Another important challenge is when nuisance parameters are present.  Our interest is in the heritability coefficient, a function of the full parameter $(\beta, \sigma^2)$, so we need to marginalize over the nuisance parameters.  The next section marginalizes over $\beta$ using standard techniques; further marginalization will be carried out in Section~\ref{SS:construction}.

\subsection{Marginalizing out the fixed effect}
\label{SS:fixed.effect}

Start with the linear mixed model \eqref{eq:mixed.effects}.  Following the setup in \citet{e.hannig.iyer.2008}, let $K$ be a $n \times (n-p)$ matrix such that $KK^\top = I_n - X(X^\top X)^{-1}X^\top$ and $K^\top K = I_{n-p}$.  Next, let $B = (X^\top X)^{-1}X^\top$.  Then $y \mapsto (K^\top y, B y)$ is a one-to-one mapping.  Moreover, the distribution of $K^\top Y$ depends on $(\avar, \evar)$ only, and the distribution of $B Y$ depends on $(\beta,\sigma^2)$, with $\beta$ a location parameter.  In particular, from \eqref{eq:Y.marginal}, we get 
\[ K^\top Y \sim \nm_{n-p}(0, \evar I_{n-p} + \avar G) \quad \text{and} \quad BY \sim \nm_p(\beta, C_\sigma), \]
where $G = K^\top Z A Z^\top K$ and $C_\sigma$ is a $p \times p$ covariance matrix of a known form that depends on $\sigma^2=(\avar, \evar)$; its precise form is not important.  From this point, the general theory in \citet{immarg} allows us to marginalize over $\beta$ by simply deleting the $BY$ component.  Therefore, a marginal association for $(\avar, \evar)$ is
\[ K^\top Y = (\evar I_{n-p} + \avar G)^{1/2} \, U_2, \quad U_2 \sim \nm_{n-p}(0,I_{n-p}). \]
This marginalization reduces the auxiliary variable dimension from $n$ to $n-p$.  

In the marginal association for $(\avar, \evar)$ above, there are $n-p$ auxiliary variables but only two parameters.  Classical results on sufficient statistics in mixed effects model that will facilitate further dimension reduction.  For the matrix $G$ defined above, let $\lambda_1 > \cdots > \lambda_L \geq 0$ denote the (distinct) ordered eigenvalues with multiplicities $r_1,\ldots,r_L$, respectively.  Let $P=[P_1,\ldots,P_L]$ be a $(n-p) \times (n-p)$ orthogonal matrix such that $P^\top G P$ is diagonal with eigenvalues $\{\lambda_\ell: \ell=1,\ldots,L\}$, in their multiplicities, on the diagonal.  For $P_\ell$, a $(n-p) \times r_\ell$ matrix, define 
\[ S_\ell = Y^\top K P_\ell P_\ell^\top K^\top Y, \quad \ell=1,\ldots,L. \]
\citet{olsen.seely.birkes.1976} showed that $(S_1,\ldots,S_L)$ is a minimal sufficient statistic for $(\avar, \evar)$.  Moreover, the distribution of $(S_1,\ldots,S_L)$ is determined by 
\begin{equation}
\label{eq:baseline}
S_\ell = (\lambda_\ell \avar + \evar) V_\ell, \quad V_\ell \sim \chisq(r_\ell), \quad \text{independent}, \quad \ell=1,\ldots,L. 
\end{equation}
This reduces the auxiliary variable dimension from $n-p$ to $L$.  We take \eqref{eq:baseline} as our ``baseline association.''  Even in this reduced baseline association, there are $L$ auxiliary variables but only two parameters, which means there is room for even further dimension reduction.  The next section shows how to reduce to a scalar auxiliary variable when the parameter of interest is the scalar heritability coefficient.

\section{Inferential model for heritability}
\label{S:marginal}

\subsection{Association}
\label{SS:construction}

For the moment, it will be convenient to work with the variance ratio, $\psi = \avar / \evar$.  Since $\psi=\rho / (1 - \rho)$ is a one-to-one function of $\rho$, the two parametrizations are equivalent.  Rewrite the baseline association \eqref{eq:baseline} as 
\begin{equation}
\label{eq:e.psi.assoc}
S_\ell = \evar(\lambda_\ell \psi + 1) V_\ell, \quad V_\ell \sim \chisq(r_\ell), \quad \text{independent}, \quad \ell=1,\ldots,L.
\end{equation}
If we make the following transformations, 
\begin{align*}
X_\ell & = (S_\ell/r_\ell) \bigr/ (S_L / r_L), \quad \ell=1,\ldots,L-1, \qquad X_L = S_L, \\
U_\ell & = (V_\ell / r_\ell) \bigr/ (V_L / r_L), \quad \ell=1,\ldots,L-1, \qquad U_L = V_L.
\end{align*}
then the association \eqref{eq:e.psi.assoc} becomes
\[ X_\ell = \frac{\lambda_\ell \psi + 1}{\lambda_L \psi + 1} \, U_\ell, \quad \ell=1,\ldots,L-1, \qquad X_L = \evar(\lambda_L \psi + 1) U_L. \]
Since for every $(X, U, \psi)$, there exists a $\evar$ that solves the right-most equation, it follows from the general theory in \citet{immarg} that a marginal association for $\psi$ is obtained by deleting the component above involving $\evar$.  In particular, a marginal association for $\psi$ is  
\[ X_\ell = \frac{\lambda_\ell \psi + 1}{\lambda_L \psi + 1} U_\ell, \quad \ell=1,\ldots,L-1. \]
If we write 
\[ f_\ell(\rho) = \frac{1 + \rho(\lambda_\ell - 1)}{1 + \rho(\lambda_L-1)}, \quad \ell=1,\ldots,L-1, \]
then we get a marginal association for $\rho$ of the form 
\begin{equation}
\label{eq:rho.assoc.1}
X_\ell = f_\ell(\rho) U_\ell, \quad \ell=1,\ldots,L-1. 
\end{equation}
Marginalization reduces the auxiliary variable dimension by 1.  Further dimension reduction will be considered next.  Note that the new auxiliary variable $U$ is a multivariate F-distributed random vector \citep[e.g.,][]{amos.bulgren.1972}.   

Here we construct a local conditional IM for $\rho$ as described in \citet{imcond}.  Select a fixed value $\rho_0$; more on this choice in Section~\ref{SS:pi}.  To reduce the dimension of the auxiliary variable $U$, in \eqref{eq:rho.assoc.1}, from $L-1$ to $1$, we construct two pairs of functions---one pair, $(T,H_0)$, on the sample space, and one pair, $(\tau, \eta_0)$, on the auxiliary variable space.  We insist that $x \mapsto (T(x),H_0(x))$ and $u \mapsto (\tau(u), \eta_0(u))$ are both one-to-one, and $H_0=H_{\rho_0}$ and $\eta_0=\eta_{\rho_0}$ are allowed to depend on the selected $\rho_0$.  The goal is to find $\eta_0$ which is, in a certain sense, not sensitive to changes in the auxiliary variable.  

Write the association \eqref{eq:rho.assoc.1} so that $u$ is a function of $x$ and $\rho$, i.e., 
\[ u_\ell(x,\rho) = x_\ell / f_\ell(\rho), \quad \ell=1,\ldots,L-1. \]
We want to choose the function $\eta_0$ such that the partial derivative of $\eta_0(u(x,\rho))$ with respect to $\rho$ vanishes at $\rho=\rho_0$.  By the chain rule, we have 
\[ \frac{\del \eta_0(u(x,\rho))}{\del\rho} = \frac{\del \eta_0(u)}{\del u} \Bigr|_{u=u(x,\rho)} \frac{\del u(x,\rho)}{\del \rho}, \]
so our goal is to find $\eta_0$ to solve the following partial differential equation:
\[ \frac{\del \eta_0(u)}{\del u} \Bigr|_{u=u(x,\rho)} \frac{\del u(x,\rho)}{\del \rho} = 0, \quad \text{at $\rho=\rho_0$}; \]
here $\del u / \del\rho$ is a $(L-1) \times 1$ vector and $\del\eta_0/\del u$ is a $(L-2) \times (L-1)$ matrix of rank $L-2$.  Toward solving the partial differential equation, first we get that the partial derivative of $u_\ell(x,\rho)$ with respect to $\rho$ satisfies
\[ \frac{\del u_\ell(x,\rho)}{\del\rho} = -\frac{f_\ell'(\rho)}{f_\ell(\rho)^2} x_\ell = -g(\rho) \, u_\ell(x,\rho), \]
where $g(\rho) = \{\frac{\del}{\del\rho}\log f_1(\rho),\ldots,\frac{\del}{\del\rho} \log f_{L-1}(\rho)\}^\top$.  This simplifies the relevant partial differential equation to the following:
\begin{equation}
\label{eq:pde}
\frac{\del \eta_0(u)}{\del u} \Bigr|_{u=u(x,\rho)} \,\diag\{ u(x,\rho) \} \, g(\rho) = 0, \quad \text{at $\rho=\rho_0$}, 
\end{equation}
where $\diag(a)$ is a diagonal matrix constructed from a vector $a$.  The method of characteristics \citep[e.g.,][]{polyanin2002} for solving partial differential equations identifies a logarithmic function of the form
\begin{equation}
\label{eq:eta0}
\eta_0(u)^\top = (\log u_1,\ldots, \log u_{L-1}) M_0^\top, 
\end{equation}
where $M_0=M_{\rho_0}$ is a $(L-2) \times (L-1)$ matrix with rows orthogonal to $g(\rho)$ at $\rho=\rho_0$.  For example, since the matrix that projects to the orthogonal complement to the column space of $g(\rho_0)$ has rank $L-2$, we can take $M_0$ to be a matrix whose $L-2$ rows form a basis for that space.  For $M_0$ defined in this way, it is easy to check that $\eta_0$ in \eqref{eq:eta0} is indeed a solution to the partial differential equation \eqref{eq:pde}.  

We have completed specification of the $\eta_0$ function; it remains to specify $\tau$ and $(T,H_0)$.  The easiest to specify next is $H_0(x)$, the value of $\eta(u(x,\rho_0))$, as a function of $x$:
\[ H_0(x)^\top = \Bigl( \log\frac{x_1}{f_1(\rho_0)}, \ldots, \log\frac{x_{L-1}}{f_{L-1}(\rho_0)} \Bigr) M_0^\top. \]
As we describe below, the goal is to condition on the observed value of $H_0(X)$.  

Next, we define $T$ and $\tau$ to supplement $H_0$ and $\eta_0$, respectively.  In particular, take a $(L-1)$-vector $w(\rho)$ which is not orthogonal to $g(\rho)$ at $\rho=\rho_0$.  It is easy to check that the entries in $g(\rho)$ are strictly positive for all $\rho$.  Therefore, we can take $w(\rho)$ independent of $\rho$, e.g., $w(\rho) \equiv 1_{L-1}$, is not orthogonal to $g(\rho)$.  Now set
\begin{equation}
\label{eq:log.linear}
\begin{pmatrix} \tau(u) \\ \eta_0(u) \end{pmatrix} = \begin{pmatrix} 1_{L-1}^\top \\ M_0 \end{pmatrix} \begin{pmatrix} \log u_1 \\ \vdots \\ \log u_{L-1} \end{pmatrix}. 
\end{equation}
This is a log-linear transformation, and the linear part is non-singular, so this is a one-to-one mapping.  Finally, we take $T$ as
\[ T(x) = \sum_{\ell=1}^{L-1} \log x_\ell.  \]
Since $(T(x),H_0(x))$ is log-linear, just like $(\tau(u), \eta_0(u))$, it is also one-to-one.  

We can now write the conditional association for $\rho$.  For the given $\rho_0$, the mapping $x \mapsto (T(x),H_0(x))$ describes a split of our previous association \eqref{eq:rho.assoc.1} into two pieces: 
\[ \sum_{\ell=1}^{L-1} \log X_\ell = \sum_{\ell=1}^{L-1} \log f_\ell(\rho) + \sum_{\ell=1}^{L-1} \log U_\ell \quad \text{and} \quad H_0(X)=\eta_0(U). \]
The first piece carries direct information about $\rho$.  The second piece plays a conditioning role, correcting for the fact that some information was lost in reducing the $(L-1)$-dimensional $X$ to a one-dimensional $T(X)$.  To complete the specification of the conditional association, write $\phi(\rho) = \sum_{\ell=1}^{L-1} \log f_\ell(\rho)$ and $V = \sum_{\ell=1}^{L-1} \log U_\ell$.  Then we have 
\begin{equation}
\label{eq:rho.assoc.2}
T(X) = \phi(\rho) + V, \quad V \sim \prob_{V|h_0, \rho_0}, 
\end{equation}
where $\prob_{V|h_0, \rho_0}$ is the conditional distribution of $\tau(U)$, given that $\eta_0(U)$ equals to the observed value $h_0$ of $H_0(X)$.   To summarize, \eqref{eq:rho.assoc.2} completes the association step that describes the connection between observable data, unknown parameter of interest, and unobservable auxiliary variables.  Of particular interest is that this association involves only a one-dimensional auxiliary variable compared to the association \eqref{eq:e.psi.assoc} obtained from the minimal sufficient statistics that involves an $L$-dimensional auxiliary variable.  This dimension reduction will come in handy for the choice of predictive random set in the following section.  The price we paid for this dimension reduction was the choice of a particular localization point $\rho_0$.  In Section~\ref{SS:pi} we employ a trick to side-step this issue when the goal is, as in this paper, to construct confidence intervals.

\subsection{Predictive random sets}
\label{SS:prs}

Having reduced the auxiliary variable to a scalar in \eqref{eq:rho.assoc.2}, the choice of an efficient predictive random set is now relatively simple.  Though there is an available theory of optimal predictive random sets \citep{imbasics}, here we opt for simplicity; in particular, we propose a default predictive random set that is theoretically sound and computational and intuitively simple.  

Consider the following version of what \citet{imbasics} call the ``default'' predictive random set:
\begin{equation}
\label{eq:default.prs}
\S = \{v: |v-\mu_0| \leq |V-\mu_0|\}, \quad V \sim \prob_{V|h_0, \rho_0}. 
\end{equation}
This $\S$, with distribution $\prob_{\S|h_0,\rho_0}$, is a random interval, centered at the mean $\mu_0$ of the conditional distribution $\prob_{V|h_0, \rho_0}$.  One key feature of this predictive random set is that it is nested, i.e., for any two distinct realizations of $\S$, one is a subset of the other.  The second key feature is a calibration of the predictive random set distribution with the underlying distribution of the auxiliary variable.  Following \citet{randset}, define the contour function 
\[ \gamma_\S(v) = \prob_{\S|h_0,\rho_0}(\S \ni v), \]
which represents the probability that the predictive random set contains the value $v$ of the auxiliary variable.  We shall require that 
\begin{equation}
\label{eq:prs.valid}
\prob_{V|h_0,\rho_0}\{\gamma_\S(V) \geq 1-\alpha\} \leq \alpha, \quad \forall \; \alpha \in (0,1). 
\end{equation}
For the default predictive random set in \eqref{eq:default.prs}, it is easy to check that 
\[ \gamma_\S(v) = \prob_{V|h_0,\rho_0}\{|V-\mu_0| \geq |v-\mu_0|\} = 1 - F_{h_0,\rho_0}(|v-\mu_0|), \]
where $F_{h_0,\rho_0}$ is the distribution function of $|V-\mu_0|$ for $V \sim \prob_{V|h_0,\rho_0}$.  From the construction above, it is clear that it is a continuous distribution.  Then, $|V-\mu_0|$ is a continuous random variable, so $\gamma_\S(V)$ is uniformly distributed on $(0,1)$.  Therefore, \eqref{eq:prs.valid} holds for the default predictive random set $\S$.  Results on optimal predictive random sets are available \citep{imbasics}, but here, again, our focus is on simplicity.  See Section~\ref{S:discuss}.

\subsection{Plausibility intervals}
\label{SS:pi}

Here we combine the association in Section~\ref{SS:construction} with the predictive random set described above to produce a plausibility function for inference about $\rho$.  In general, a plausibility function is a data-dependent mapping that assigns, to each assertion about the parameter, a value in $[0,1]$, with the interpretation that small values suggest the assertion is not likely to be true, given the data; see \citet{imbasics}.  For simplicity, we focus only on the collection of singleton assertions, i.e., $\{\rho=r\}$ for $r \in [0,1]$.  These are also the relevant assertions for constructing interval estimates based on the IM output.  

Let $X=x$ be the observations in \eqref{eq:rho.assoc.1}.  The association step in Section~\ref{SS:construction} yields a data-dependent collection of sets indexed by the auxiliary variable.  In particular, write $R_x(v) = \{\rho: T(x)=\phi(\rho) + v\}$, a set-valued function of $v$.  These sets are combined with the predictive random set in Section~\ref{SS:prs} to get an enlarged $x$-dependent random set:
\begin{equation}
\label{eq:union}
R_x(\S) = \bigcup_{v \in \S} R_x(v). 
\end{equation}
Now, for a given assertion $\{\rho=r\}$, we compute the plausibility function, 
\[ \pl_{x|h_0,\rho_0}(r) = \prob_{\S|h_0,\rho_0}\{R_x(\S) \ni r\}, \]
the probability that the random set $R_x(\S)$ contains the asserted value $r$ of $\rho$.  A simple calculation shows that, in this case with singleton assertions, we have 
\[ \pl_{x|h_0,\rho_0}(r) = \gamma_\S\bigl(T(x) - \phi(r)\bigr) = 1 - F_{h_0,\rho_0}(|T(x)-\phi(r)-\mu_0|), \]
where $F_{h_0,\rho_0}$ is defined in Section~\ref{SS:prs}.  The above display shows that the plausibility function can be expressed directly in terms of the distribution of the predictive random set, without needing to go through the construction of $R_x(\S)$ as in \eqref{eq:union}.    

We pause here to answer a question that was left open from Section~\ref{SS:construction}, namely, how to choose the localization point $\rho_0$.  Following \citet{imcond}, we propose here to choose $\rho_0$ to match the value of $\rho$ specified by the singleton assertion.  That is, we propose to let the localization point depend on the assertion.  All the elements in the plausibility function above with a 0 subscript, denoting dependence on $\rho_0$, are changed in an obvious way to get a new plausibility function
\begin{equation}
\label{eq:pl.fun}
\pl_{x|h_\rho,\rho}(\rho) = 1 - F_{h_\rho,\rho}(|T(x)-\phi(\rho)-\mu_\rho|). 
\end{equation}
We treat this as a function of $\rho$ to be used for inference.  In particular, we can construct a $100(1-\alpha)$\% plausibility interval for $\rho$ as follows:
\begin{equation}
\label{eq:pl.int}
\Pi_\alpha(x) = \{\rho: \pl_{x|h_\rho, \rho}(\rho) > \alpha\}.
\end{equation}
The plausibility function, and the corresponding plausibility region, are easy to compute, as we describe in Section~\ref{SS:compute}.  Moreover, the calibration \eqref{eq:prs.valid} of the predictive random set leads to exact plausibility function-based confidence intervals, as we now show.   

We need some notation for the sampling distribution of $X$, given all the relevant parameters.  Recall that the distribution of $X$ actually depends on $(\avar, \evar)$ or, equivalently, $(\rho, \evar)$.  The error variance $\evar$ is a nuisance parameter, but $\evar$ still appears in the sampling model for $X$.  We write this sampling distribution as $\prob_{X|\rho, \evar}$.  

\begin{theorem}
\label{thm:valid}
Take the association \eqref{eq:rho.assoc.2} and the default predictive random set $\S$ in \eqref{eq:default.prs}.  Then for any $\rho$, any value $h_\rho$ of $H_\rho$, and any $\evar$, the plausibility function satisfies 
\begin{equation}
\label{eq:valid.thm}
\prob_{X|\rho, \evar}\bigl\{ \pl_{X|h_\rho, \rho}(\rho) \leq \alpha \mid H_\rho(X)=h_\rho \bigr\} = \alpha, \quad \forall \; \alpha \in (0,1). 
\end{equation}
\end{theorem}

\begin{proof}
For given $(\evar, \rho)$, if $h_\rho$ is the value of $H_\rho(X)$, then it follows from the conditional distribution construction that the plausibility function in \eqref{eq:pl.fun}, as a function of $T(X)$ with $X \sim \prob_{X|\rho, \evar}$, is $\unif(0,1)$.  Then the equality in \eqref{eq:valid.thm} follows immediately.  
\end{proof}

Averaging the left-hand side of \eqref{eq:valid.thm} over $h_\rho$, with respect to the distribution of $H_\rho(X)$, and using iterated expectation gives the following unconditional version of Theorem~\ref{thm:valid}.

\begin{corollary}
\label{cor:valid}
Under the conditions of Theorem~\ref{thm:valid}, for any $(\rho, \evar)$, 
\[ \prob_{X|\rho, \evar}\bigl\{ \pl_{X|H_\rho(X), \rho}(\rho) \leq \alpha \bigr\} = \alpha, \quad \forall \; \alpha \in (0,1). \]
\end{corollary}

Since we have proper calibration of the plausibility function, both conditionally and unconditionally, coverage probability results for the plausibility interval \eqref{eq:pl.int} are also available.  This justifies our choice to call $\Pi_\alpha(x)$ a $100(1-\alpha)$\% plausibility interval, i.e., the frequentist coverage probability of $\Pi_\alpha$ is exactly $1-\alpha$.  

\begin{corollary}
\label{cor:coverage}
The coverage probability of $\Pi_\alpha(X)$ in \eqref{eq:pl.int} is exactly $1-\alpha$.
\end{corollary}

\section{Numerical results}
\label{S:numerical}

\subsection{Implementation}
\label{SS:compute}

Evaluation of the plausibility function in \eqref{eq:pl.fun} requires the distribution function $F_{h_\rho,\rho}$ of $|V-\mu_\rho|$ corresponding to the conditional distribution $\prob_{V|h_\rho,\rho}$ of $V=\tau(U)$, given $\eta_\rho(U) = H_\rho(X)$.  This conditional distribution is not of a convenient form, so numerical methods are needed.  For $\rho$ fixed, since the transformation \eqref{eq:log.linear} from $U$ to $(\tau(U), \eta_\rho(U))$ is of a log-linear form, and the density function of $U$ can be written in closed-form, we can evaluate the joint density for $(\tau(U),\eta_\rho(U))$ and, hence, the conditional density of $V=\tau(U)$.  Numerical integration is used to evaluate the normalizing constant, the mean $\mu_\rho$, and the distribution function $F_{h_\rho,\rho}$.  R code is available at \url{www.math.uic.edu/~rgmartin}.

\subsection{Simulation results}
\label{SS:simulation}

In this section, we consider a standard one-way random effects model, i.e., 
\[ y_{ij} = \mu + \alpha_i + \eps_{ij}, \quad i=1,\ldots,a, \quad j=1,\ldots,n_i, \]
where $\alpha_1,\ldots,\alpha_a$ are independent with common distribution $\nm(0,\avar)$, and the $\eps_{ij}$s are independent with common distribution $\nm(0,\evar)$; the $\alpha_i$s and $\eps_{ij}$s are also mutually independent.  Our goal is to compare the proposed IM-based plausibility intervals for $\rho$ with the confidence intervals based on several competing methods.  Of course, the properties of the various intervals depend on the design, in this case, the within-group sample sizes $n_1,\ldots,n_a$, and the values of $(\avar, \evar)$.  Our focus here is on cases with small sample sizes, namely, where the total sample size $n=n_1+\cdots+n_a$ is fixed at 15.  The three design patterns $(n_1,\ldots,n_a)$ considered are: $(1,1,1,1,1,10)$, $(2,4,4,5)$, and $(2,3,10)$.  The nine $(\avar, \evar)$ pairs considered are: $(0.05, 10)$, $(0.1, 10)$, $(0.5, 10)$, $(1,10)$, $(0.5, 2)$, $(1, 1)$, $(2, 0.5)$, $(5, 0.2)$, and $(10, 0.1)$. Without loss of generality, we set $\mu=0$.  

For each design pattern and pair of $(\avar, \evar)$, 1000 independent data sets were generated and 95\% two-sided interval estimates for $\rho$ were computed based on the exact method of \citet{burch.iyer.1997}, the fiducial method of \citet{e.hannig.iyer.2008}, and the proposed IM method.  Empirical coverage probabilities and average length of the confidence interval under each setting were compared to investigate the performance of each method.  Besides these three methods, we also implemented Bayesian and profile likelihood approaches.  The three aforementioned methods all gave intervals with better coverage and length properties than the Bayesian method, and the profile likelihood method was very unstable with small sample sizes, often having very high coverage with very wide intervals or very low coverage with very narrow intervals.  So, these results are not reported.  

A summary of the simulation results is displayed in Figure~\ref{fig:sim1}.  Panel~(a) displays the coverage probabilities, and Panel~(b) displays the relative length difference, which is defined as the length of the particular interval minus the length of the IM interval, scaled by the length of the IM interval.  As we expect from Corollary~\ref{cor:coverage}, the IM plausibility intervals have coverage at the nominal 95\% level.  We also see that the IM intervals tend to be shorter than the fiducial and Burch--Iyer confidence intervals.  The fiducial intervals have coverage probability exceeding the nominal 95\% level, but this comes at the expense of longer intervals on average.  Overall, the proposed IM-based method performs quite well compared to these existing methods.  We also replicated the simulation study in \citet{e.hannig.iyer.2008}, which involves larger sample sizes and a broader range of imbalance, and the relative comparisons between these three methods are the same as here.

\begin{figure}[t]
\begin{center}
\subfigure[Coverage probability]{\scalebox{0.22}{\includegraphics{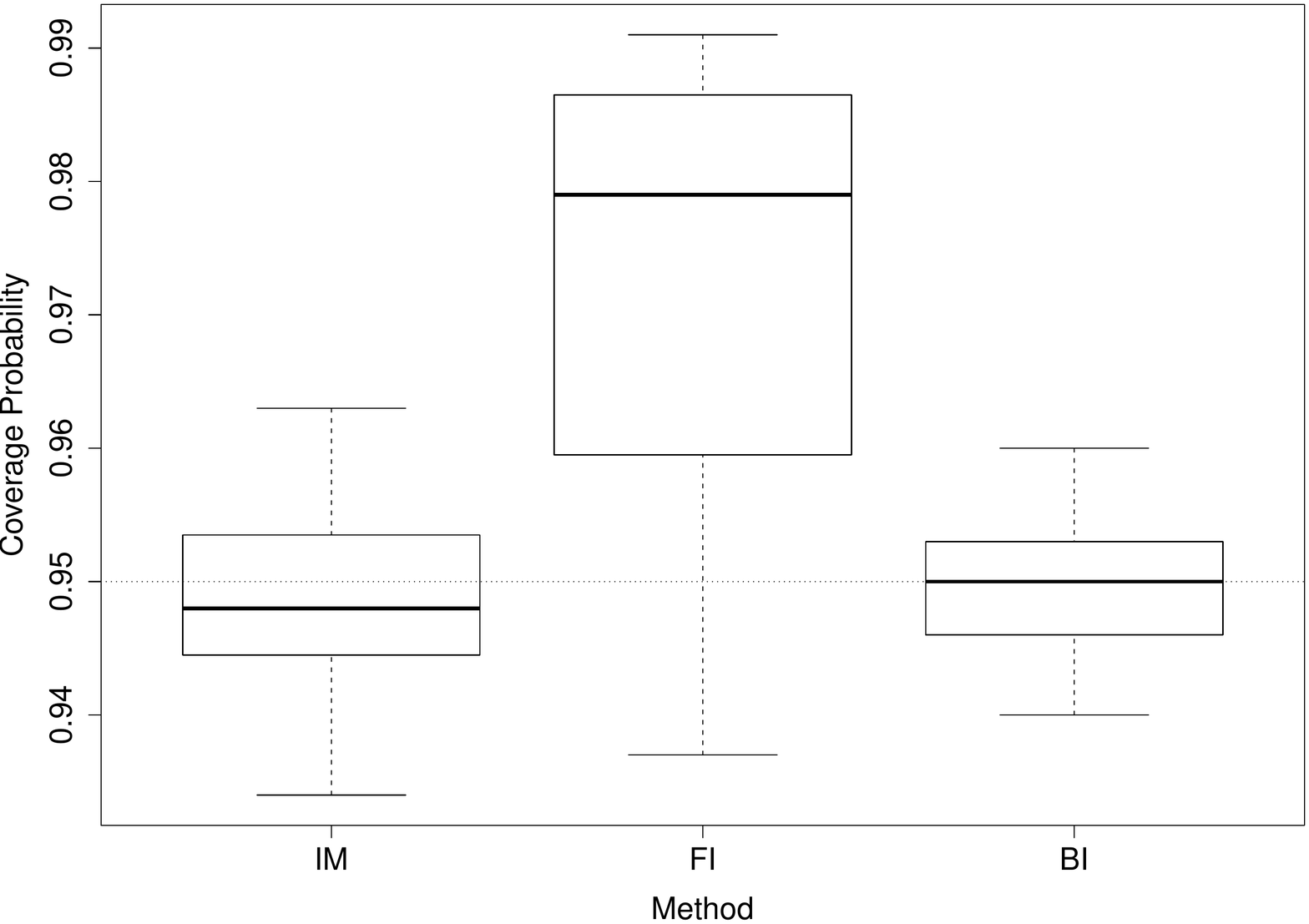}}} 
\subfigure[Relative length difference]{\scalebox{0.22}{\includegraphics{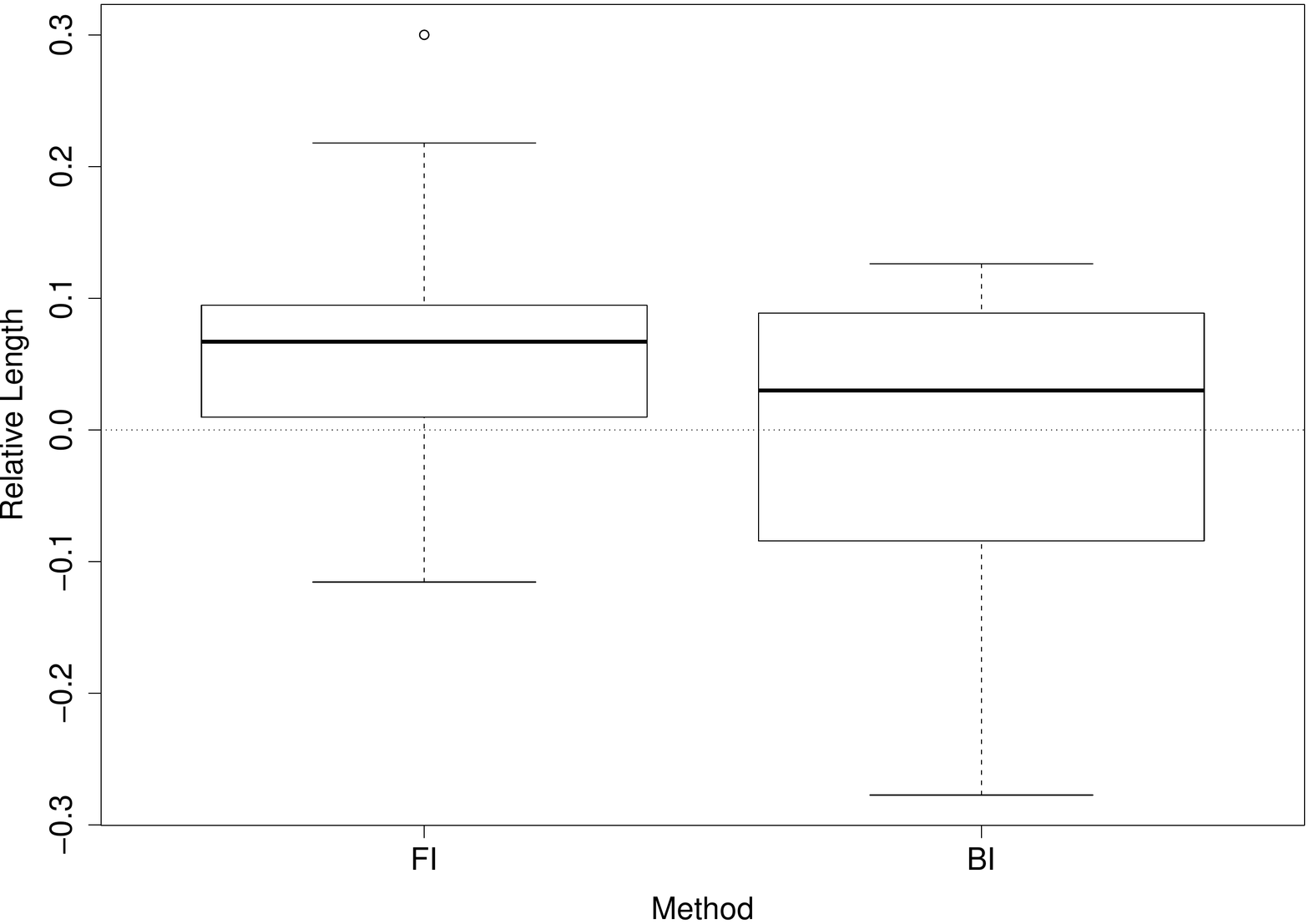}}}
\end{center}
\caption{Simulation results from Section~\ref{SS:simulation}.  BI corresponds to the exact method of \citet{burch.iyer.1997}, and FI corresponds to the fiducial method of \citet{e.hannig.iyer.2008}.}
\label{fig:sim1}
\end{figure}

\subsection{Real-data analysis}
\label{SS:real}

\begin{example}
\label{ex:assay}
Equal numbers of subjects are tested under each standard and test preparations and a blank dose under a ($2K+1$)-point symmetrical slope-ratio assay.  The response, on logarithmic scale, is assumed to depend linearly on the dose level.  A modified balanced incomplete block design with $2K'+1$ $(K'<K)$ block size is introduced by \citet{das.kulkarni.1966}.  The $i$th dose levels for standard and test preparations are represented by $s_i$ and $t_i$, $i=1,\ldots,K$.  Under this design, the dose will be equally spaced and listed in ascending order.  A balanced incomplete block design with $K$ doses of the standard preparation inside $K'$ blocks is constructed and used as the basic design.  Then a modified design is constructed by adding a blank dose and $K'$ doses of the test preparation into every block, under the rule that dose $t_i$ should accompany $s_i$ in every blocks.  The model developed by Das and Kulkarni can be written as 
\[ y_{ijm}=\mu+\beta_jx_{ij}+\alpha_m+ \eps_{ijm},\quad i \in \{s,t,c\}, \quad j=1,\ldots,k,\quad m =1,\ldots,b \]
where $y_{sjm}$, $y_{tjm}$, $y_{cjm}$ represent observation response in $m$th block for $j$th dose of standard preparation, test preparation and blank dose; $x_{sj}$ and $x_{tj}$ represent $j$th dose level for standard and test preparation; $x_{cj}$ is zero by default, $\alpha_m$ denotes $m$th block effect; and $\eps_{ijm}$ denotes independent random errors with common distribution $\nm(0,\evar)$.  We consider random block effects and assume that $\alpha_m$ are independent with common distribution $\nm(0,\avar)$.  Independence of $\alpha_m$ and $\eps_{ijm}$ is also assumed.

We analyze data coming from a nine-point slope-ratio assay on riboflavin content of yeast, with two replications in each dose; see Table~2 in \citet{e.hannig.iyer.2008} for the design and data.  For this design, we have $L=3$ distinct eigenvalues, namely, 4.55, 1, and 0, with multiplicities 1, 1, and 10, respectively.  A plot of the plausibility function for $\rho$ is shown in Figure~\ref{fig:pl}(a).  The function exceeds 0.2 at $\rho=0$, which implies that 90\% and 95\% plausibility intervals include zero.  The left panel of Table~\ref{table:real} shows the 90\% and 95\% interval estimates for $\rho$ based on the Burch--Iyer, fiducial, and IM methods.  In this case, the IM intervals are provably exact and shorter.  
\end{example}

\begin{figure}[t]
\begin{center}
\subfigure[Example~\ref{ex:assay}]{\scalebox{0.4}{\includegraphics{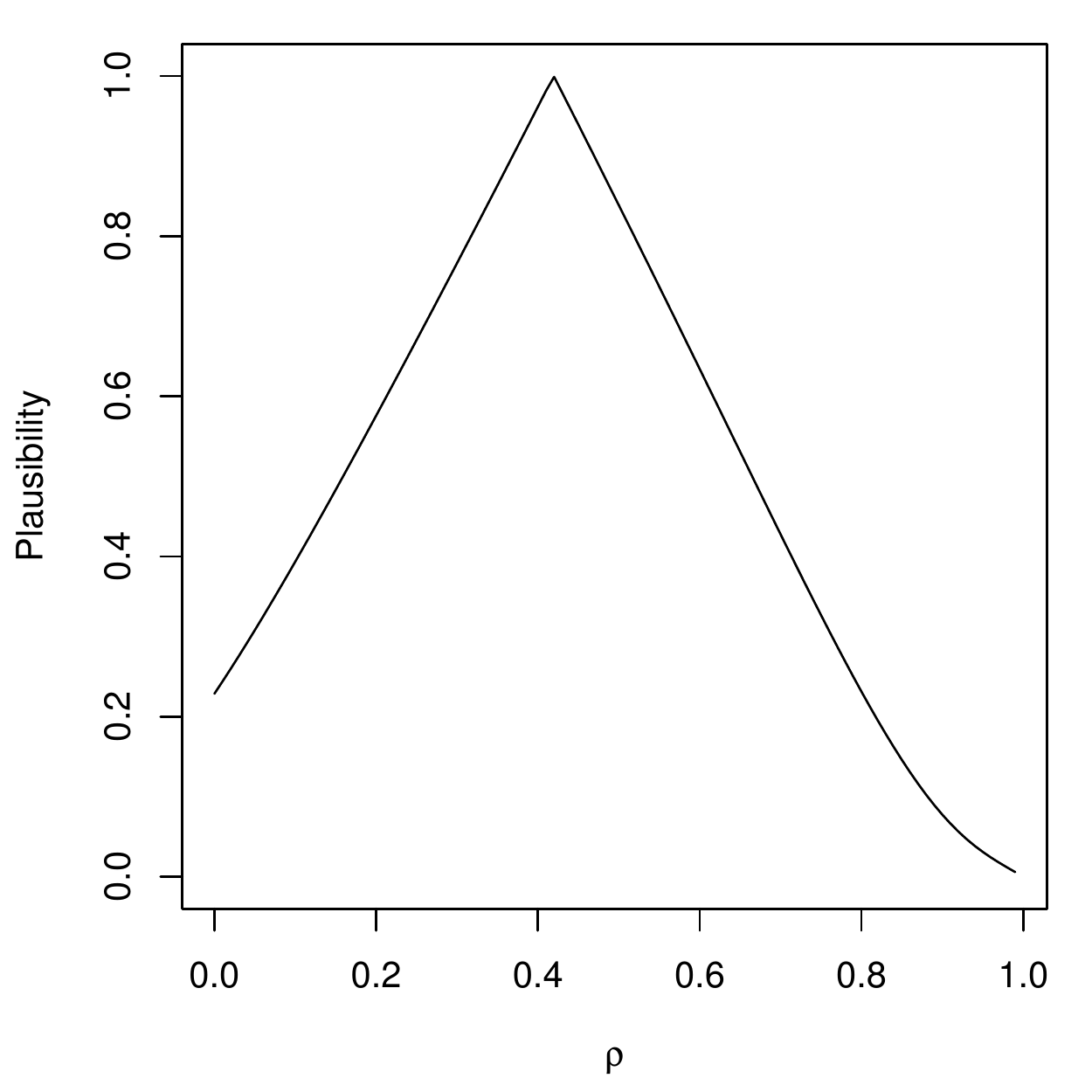}}}
\subfigure[Example~\ref{ex:lamb}]{\scalebox{0.4}{\includegraphics{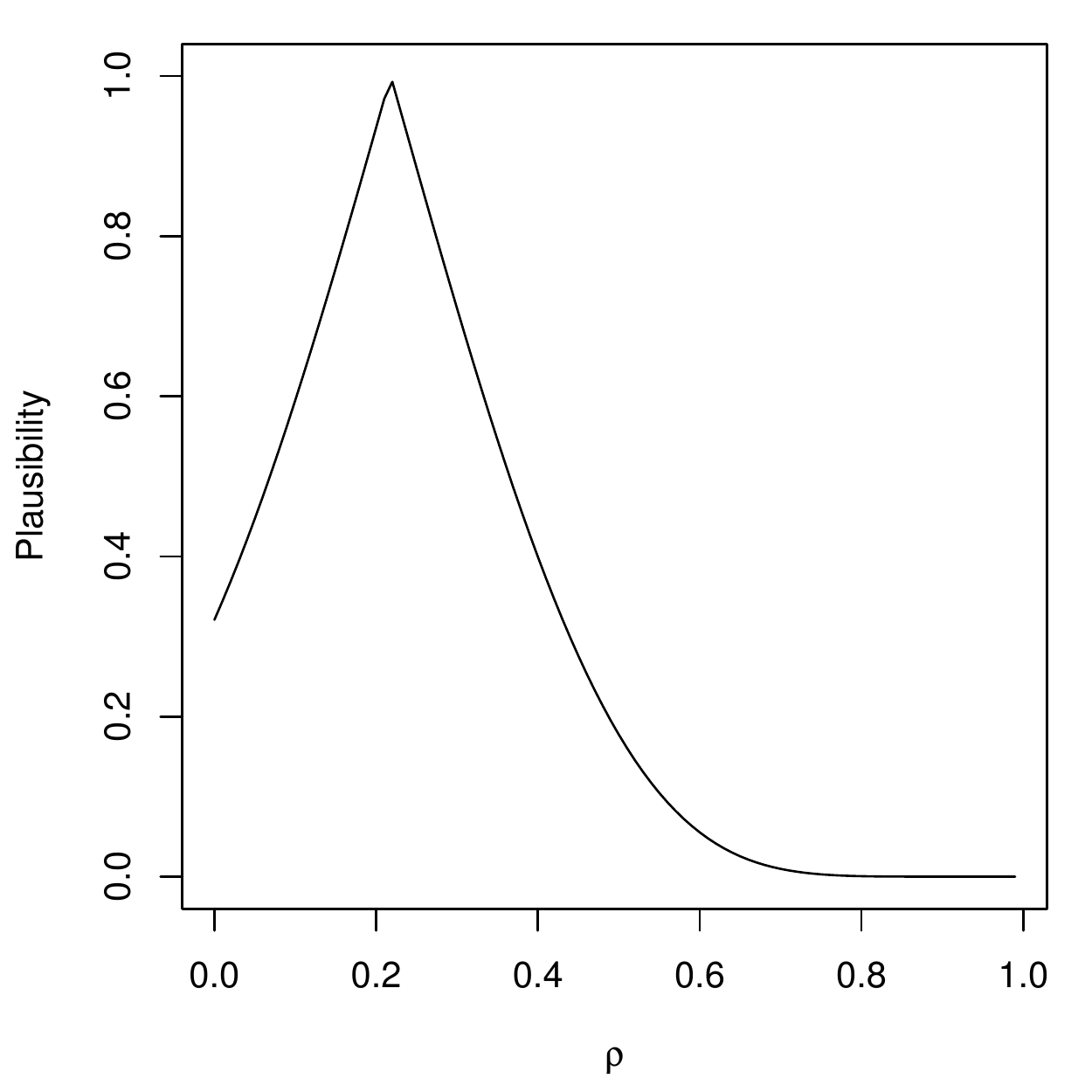}}}
\end{center}
\caption{Plausibility functions for $\rho$ in the two examples.}
\label{fig:pl}
\end{figure}

\begin{table}
\begin{center}
\begin{tabular}{ccccccc}
\hline
& & \multicolumn{2}{c}{Example~\ref{ex:assay}} & & \multicolumn{2}{c}{Example~\ref{ex:lamb}} \\
Method & & 90\% & 95\% & & 90\% & 95\% \\
\hline 
Burch--Iyer & & 0.913 & 0.956 & & 0.567 & 0.615  \\
Fiducial & & 0.916 & 0.957 & & 0.466 & 0.530 \\
IM & & 0.881 & 0.924 & & 0.554 & 0.597 \\ 
\hline
\end{tabular}
\caption{Upper bounds on the interval estimates for $\rho$ (lower bounds are all zero) based on the three methods in the two real-data examples.}
\label{table:real}
\end{center}
\end{table}

\ifthenelse{1=1}{}{

# Results for new IM with numerical integration:

> o <- assay.example(alpha=0.05, plot.pl=FALSE)
[1] 0.0000000 0.9264546
> o <- assay.example(alpha=0.1, plot.pl=FALSE)
[1] 0.0000000 0.8816783
> o <- lamb.example(alpha=0.05, plot.pl=FALSE)
[1] 0.0000000 0.6071329
> o <- lamb.example(alpha=0.1, plot.pl=FALSE)
[1] 0.0000000 0.5539082

# Results for REML (using observed information for variance):

> o <- assay.example(alpha=0.05, plot.pl=FALSE)
[1] 0.0000000 0.9264546
> o <- assay.example(alpha=0.1, plot.pl=FALSE)
[1] 0.0000000 0.8816783
> o <- lamb.example(alpha=0.05, plot.pl=FALSE)
[1] -0.1336466  0.5226253
> o <- lamb.example(alpha=0.1, plot.pl=FALSE)
[1] -0.0808910  0.4698697

}

\begin{example}
\label{ex:lamb}
\citet{harville.fenech.1985} analyzed data on birth weights of lambs.  These data consist of the weights information at the birth of 62 single-birth male lambs, and were collected from three selection lines and two control lines.  Each lamb was the offspring of one of the 23 rams and each lamb had a distinct dam.  Age of the dam was also recorded and separated into three categories, numbered 1 (1--2 years), 2 (2--3 years), and 3 (over 3 years).   
A linear mixed model for these data is
\[ y_{ijkl}=\mu+\beta_i+\pi_j+\alpha_{jk}+\eps_{ijkl}, \]
where $y_{ijkl}$ represents the weight of the $l$th offspring of the $k$th sire in the $j$th population lines and of a dam in the $i$th age category; $\beta_i$ represents $i$th level age effect; $\pi_j$ represents the $j$th line effects; $\alpha_{jk}$ denotes random sire effects and are assumed to be independently distributed as $\nm(0,\avar)$; and random errors denoted by $\eps_{ijkl}$ is supposed to be independently distributed as $\nm(0,\evar)$.  Furthermore, the $\alpha_{jk}$s and $\eps_{ijkl}$s are assumed to be independent.  In this case, $L=18$, $\lambda_1=5.09$, $\lambda_L=0$ and $r_L=37$; all non-zero eigenvalues have multiplicity 1 except $\lambda_8=2$ with multiplicity 2.  A plot of the plausibility function for $\rho$ is shown in Figure~\ref{fig:pl}(b).  As in the previous example, the plausibility function is positive at $\rho=0$, which means that plausibility intervals with any reasonable level will contain $\rho=0$.  We also used each of the three methods considered above to compute 90\% and 95\% interval estimates for $\rho$.  The results are shown in the right panel of Table~\ref{table:real}.  In this case, IM gives a shorter interval compared to Burch--Iyer.  The fiducial interval, however, is shorter than both exact intervals.  We expect the IM interval to be most efficient, so we explore the relative performance a bit further by simulating 1000 independent data sets from the fitted model in this case, i.e., with $\hat\sigma_\alpha^2=0.767$ and $\hat\sigma_\eps^2=2.763$ as the true values.  In these simulations, the fiducial and IM coverage probabilities were 0.944 and 0.954, respectively, both within an acceptable range of the nominal level, but the average lengths of the intervals are 0.488 and 0.456.  That is, the IM intervals tend to be shorter than the fiducial intervals in problems similar to this example, as we would expect.
\end{example}

\ifthenelse{1=1}{}{
Some results from "lamb weight data"

1) results of the real data
       90
BI   0.567         0.615
FI   0.466         0.530
IM  0.554         0.597
LI   0.505         0.573

2) results of simulations based on "lamb weight data"

       Coverage Probability          Length of C.I       (95
FI        94.4                                 0.488
IM       95.4                                  0.456
LI        97.5                                  0.584

}

\section{Concluding remarks}
\label{S:discuss}

The IM method proposed here gives exact confidence intervals for the heritability coefficient $\rho$, as well as the variance ratio $\psi$, and numerical results suggest increased efficiency compared to existing methods.  A question is if these same techniques can be employed for exact prior-free probabilistic inference on other quantities related to the variance components $(\avar, \evar)$.  It is well-known that, for the unbalanced design case, exact marginalization is challenging.  In the IM context, this means that the association is not ``regular'' in the sense of \citet{immarg}.  Therefore, some more sophisticated tools are needed for exact inference on, say, the individual variance components $\avar$ and $\evar$.  This application provides clear motivation for our ongoing investigations into more general IM-based marginalization strategies.  

Another important question is if the techniques presented herein can be applied in more complex and high-dimensional mixed-effect models.  In genome-wide association studies, for example, the dimensions of the problem are extremely large.  We expect that, conceptually, the techniques described here will carry over to the more complex scenario.  However, there will be computational challenges to overcome, as with all approaches \citep{kang.natgen.2010, zhou.stephens.2014}.  This, along with the incorporation of optimal predictive random sets \citep[e.g.][]{imbasics} is a focus of ongoing research.  


\section*{Acknowledgements}

The authors thank Chuanhai Liu for helpful comments and suggestions.  This work is partially supported by the U.S.~National Science Foundation, grant DMS--1208833.

\bibliographystyle{apalike}
\bibliography{/Users/rgmartin/Dropbox/Research/mybib}

\end{document}